\documentclass[11pt]{article}
\usepackage[utf8]{inputenc}

\usepackage[shortlabels]{enumitem}
\usepackage{fullpage,amsfonts,bm}
\usepackage{amsmath,amssymb,amsthm,mathtools,bbm,physics,algorithm}
\usepackage{comment}
\usepackage[noend]{algpseudocode}
\usepackage{mleftright}\mleftright
\usepackage{hyperref} 
\usepackage{algpseudocode}
\usepackage{float}

\floatstyle{boxed}
\makeatletter
\def\BState{\State\hskip-\ALG@thistlm}
\makeatother
\newfloat{algorithm}{t}{lop}
\floatname{algorithm}{Algorithm} 

\newtheorem{theorem}{Theorem}

\newtheorem{claim}[theorem]{Claim}

\newtheorem{definition}[theorem]{Definition}

\def\01{\{0,1\}}
\newcommand{\eps}{\varepsilon}

\begin{document}

\title{Improved Quantum Boosting}
\author{Adam Izdebski\thanks{Work done while a student at the University of Amsterdam. {\tt adam.izdebski1@gmail.com}}
\and
Ronald de Wolf\thanks{QuSoft, CWI and University of Amsterdam, the Netherlands. Partially supported by ERC Consolidator Grant 615307-QPROGRESS (which ended February 2019), and by the Dutch Research Council (NWO) through Gravitation-grant Quantum Software Consortium, 024.003.037, and through QuantERA project QuantAlgo 680-91-034. {\tt rdewolf@cwi.nl}}
}
\maketitle

\begin{abstract}
Boosting is a general method to convert a weak learner (which generates hypotheses that are just slightly better than random) into a strong learner (which generates hypotheses that are much better than random). Recently, Arunachalam and Maity~\cite{arunachalam&maity:qboosting} gave the first quantum improvement for boosting, by combining Freund and Schapire's  AdaBoost algorithm with a quantum algorithm for approximate counting. Their booster is faster than classical boosting as a function of the VC-dimension of the weak learner's hypothesis class, but worse as a function of the quality of the weak learner. In this paper we give a substantially faster and simpler quantum boosting algorithm, based on Servedio's SmoothBoost algorithm~\cite{servedio2003smooth}.
\end{abstract}

\section{Introduction}

\subsection{Boosting}

There has been tremendous growth in machine learning research and applications, both in practice (applying all sorts of methods on all sorts of data and seeing what works well) and in theory (computational learning theory). However, not very many ideas generated in theoretical machine learning have had a large impact on machine learning practice. One of the exceptions is \emph{boosting}, which is a simple, general, and widely applicable method to improve the generalization error of a given learning method, i.e.,  to convert a \emph{weak} learner into a \emph{strong} learner.

The set-up here is binary classification: we are trying to predict binary labels $y$ from points $x\in\mathcal{X}$. A typical case would be $\mathcal{X}=\01^n$. We are given $m$ labeled examples $(x_1,y_1),\ldots,(x_m,y_m)\in\mathcal{X}\times\{-1,1\}$ where the $x_i$s are independent and identically distributed (i.i.d.)\ according to some unknown distribution ${\mathcal{D}}$, and the binary labels are determined by some unknown \emph{target function} $f:\mathcal{X}\to\{-1,1\}$ that we are trying to learn, i.e., $y_i=f(x_i)$.
A weak learner $\mathcal{W}$ is an algorithm that can be fed a number of examples according to a specified distribution $D$ over the $m$ examples of the given sample, and that is then promised to generate a hypothesis $h:\mathcal{X}\to\{-1,1\}$ that is slightly better than random w.r.t.\ that $D$:
$$
\Pr_{x\sim D}[h(x)\neq f(x)]\leq 1/2-\gamma.
$$
Here $\gamma\in(0,1/2)$ is a small but positive number that gives the quality of the weak learner. We denote the ``cost'' of one run of $\mathcal{W}$ by $W$, and use this number also as an upper bound on the number of examples (distributed according to~$D$) that the weak learner uses.

A hypothesis with a generalization error that is just slightly better than random is not very useful by itself.
The goal of boosting is to convert the weak learner into a strong learner, which is one that produces hypotheses not only with small empirical error (i.e., w.r.t.\ the uniform distribution over the $m$ examples), but even with small generalization error w.r.t.\ the unknown target function $f:\mathcal{X}\to\{-1,1\}$ and the unknown distribution~${\mathcal{D}}$ that generated the examples:  
$$
\Pr_{x\sim \mathcal{D}}[h(x)\neq f(x)]\leq\epsilon.
$$
Here the desired upper bound $\eps$ on the final generalization error is a parameter of the strong learner. Unsurprisingly, achieving smaller $\eps$ requires a larger number of examples and larger runtime. For simplicity, in this introduction we focus on the case $\eps=1/3$ (in the body of the paper we cover the general case). 
Similarly, the smaller the initial advantage $\gamma$ is, the more work we will have to do find a hypothesis with small generalization error.\footnote{For simplicity we will assume this $\gamma$ is known to the strong learner we are trying to design, but this is not necessary: if it doesn't know 
$\gamma$, the strong learner can try exponentially decreasing guesses for $\gamma$ until it finds a hypothesis with small empirical error.}

The idea of boosting is to find a hypothesis with low empirical error by combining different runs of the weak learner on different distributions. 
Once we have a hypothesis with small empirical error on a sufficiently large set of examples, VC-theory implies that such a hypothesis will probably also have a small generalization error.

How can we find a hypothesis with small empirical error?
Because empirical error is measured w.r.t.\ the uniform distribution over $\{x_1,\ldots,x_m\}$, that will be our first distribution~$D^1$. We run the weak learner on~$D^1$, and receive a hypothesis $h_1$ that is slightly better than random w.r.t.\ the uniform distribution. The next iteration then biases the distribution away from the examples that are already well-classified, by increasing the probability of misclassified examples, yielding a new distribution $D^2$. We then run the weak learner again, to generate a hypothesis $h_2$ that is slightly better than random w.r.t.\ this new distribution, and hence hopefully better than $h_1$ on the examples that were misclassified by $h_1$. Then we bias the distribution further towards the still-misclassified examples, and so on. The intuition here is that the distributions $D^t$ ``zoom in'' on the hardest examples, the ones that are most difficult to classify correctly. After some $T$ iterations, the $T$ different weak hypotheses are combined into one hypothesis~$h$, typically by defining the latter as the sign of a linear combination $\sum_{t=1}^T\alpha_t h_t$ of the $T$ weak hypotheses $h_1,\ldots,h_T$. Surprisingly, already after a relatively small number of iterations, the resulting hypothesis will have small empirical error! Thus boosting converts the ability to generate weak hypotheses w.r.t.\ chosen distributions over the examples, into the ability to generate strong hypotheses, which have small error w.r.t.\ both the uniform distribution over the examples, and w.r.t\ the unknown target function~$f$ and distribution~$\mathcal{D}$ that generated our $m$ examples.

A number of classical boosting algorithms exist that instantiate this meta-algorithm in different ways. The most famous of these is probably Freund and Schapire's AdaBoost~\cite{freund&schapire:adaboost,freund1999short,schapire2013boosting} (short for ``adaptive boosting''), which biases the new distribution $D^{t+1}$ based on the error $\eps_t$ that $h_t$ made. It drives the empirical error all the way down to~0 (note that as soon as this error is $<1/m$ it must actually be~0). AdaBoost uses $T=O(\log(m)/\gamma^2)$ iterations. Each iteration takes time $\tilde{O}(m)$ to compute the error $\eps_t$ of $h_t$ and to update the distribution over the $m$ examples\footnote{The notation $\tilde{O}(f)$ means $O(f\cdot\mathrm{polylog}(f))$.} and runs the weak learner $\mathcal{W}$ once, at cost~$W$. This gives overall complexity
\[
\tilde{O}\left(\frac{W+m}{\gamma^2}\right).
\] 
How large should $m$ be in order to make the inference from low empirical error to low generalization error? This depends on the hypothesis space $\mathcal{H}_{weak}$ of the weak learner, in particular on its \emph{VC-dimension}~$d$ (defined in Section~\ref{sec:paclearning}). The hypothesis space $\mathcal{H}_{strong}$ of the boosting algorithm consists of all signs of linear combinations of up to $T$ elements of $\mathcal{H}_{weak}$. One can show that the VC-dimension of $\mathcal{H}_{strong}$ is $D=\tilde{O}(dT)$.
VC-theory implies that (for constant $\eps$) $m\approx D\approx dT\approx d/\gamma^2$ examples suffice to end up with generalization error $\leq \eps$ (with high probability over the choice of the sample). 
Accordingly, when re-expressed as a function of $d$ rather than $m$, the complexity of AdaBoost is
\begin{equation}\label{eq:adaboostcomplexity}
\tilde{O}\left(\frac{W}{\gamma^2}+\frac{d}{\gamma^4}\right).
\end{equation}

\subsection{Quantum boosting}

In the last few years there has been a surge in interest in possible ways in which \emph{quantum} computers might help improve machine learning (see~\cite{biamonteea:qml} for a survey of several algorithmic approaches and \cite{arunachalam:quantumlearningsurvey} for quantum learning theory).

Recently, Arunachalam and Maity~\cite{arunachalam&maity:qboosting} gave the first speed-up for boosting on a quantum computer.
Here the given sample is still the same classical sequence of $m$ labeled examples $(x_1,y_1),\ldots,(x_m,y_m)\in\mathcal{X}\times\{-1,1\}$, but these are now stored in a quantum-accessible classical memory, which means a quantum learner can query multiple examples in superposition.

The key insight of~\cite{arunachalam&maity:qboosting} is that the error $\eps_t$ of the base classifier in the $t$-th iteration of AdaBoost can be approximated faster (in time $o(m)$) using a quantum counting algorithm; this approximation is subtle because it involves both multiplicative and additive error, in different regimes for $\eps_t$.
Their method works not only for boosting classical weak learners, but also for boosting \emph{quantum} weak learners. These are fed quantum examples w.r.t.\ the distribution~$D$:
\begin{equation}\label{eq:qexample}
\sum_{i=1}^m \sqrt{D(x_i)}\ket{x_i,y_i}.
\end{equation}
If the quantum weak learner $\mathcal{W}$ expects to receive $W$ such examples, the quantum booster will have to prepare $W$ copies of this state to feed into $\mathcal{W}$.\footnote{The classical boosting literature~\cite{schapire2013boosting} distinguishes ``boosting by resampling'' and ``boosting by reweighting''. Like Arunachalam and Maity~\cite{arunachalam&maity:qboosting}, we follow ``boosting by resampling'' and explicitly prepare the $W$ quantum or classical examples (w.r.t.\ $D^t$) that the weak learner needs, rather than just modifying~$D^t$.}

The quantum version of AdaBoost of~\cite{arunachalam&maity:qboosting} uses the same number of iterations as classical AdaBoost, but improves the complexity of each iteration (at least as a function of $m$ or $d$).
Their main complexity upper bound is:
\begin{equation}\label{eq:qadaboost}
\tilde{O}\left(\frac{W^{1.5}\sqrt{d}}{\gamma^{11}}\right).
\end{equation}
Comparing with the complexity of classical AdaBoost  Eq.~\eqref{eq:adaboostcomplexity}, this gives a speed-up over classical boosting in terms of the dependence on the VC-dimension~$d$ of the weak learner's hypothesis class, but at the expense of a significant deterioration in terms of the dependence on the quality of the weak learner~$\gamma$ and a milder deterioration in terms of the weak learner's cost~$W$.

\subsection{Our results}

In this paper we give a simpler and faster quantum boosting algorithm. Instead of AdaBoost, our starting point will be Servedio's SmoothBoost algorithm~\cite{servedio2003smooth}, which we explain in Section~\ref{sec:smoothboost}. 
Servedio's motivation for smooth boosting was to deal with malicious noise (at a rate that depends on~$\gamma$) in the sample better than AdaBoost. However, SmoothBoost is also very suitable for ``quantization'' thanks to the following advantages that it has over AdaBoost:
\begin{itemize}
\item SmoothBoost doesn't need to calculate or approximate the error $\eps_t$ of $h_t$ on the $m$ examples, which means we don't need to apply approximate quantum counting for this.
\item The distributions $D^t$ that it generates are ``smooth'' (whence its name), in the sense that no example has probability much bigger than the uniform probability $1/m$. Generating quantum examples as in Eq.~\eqref{eq:qexample} is cheaper when none of the probabilities is big.
\item The weights $\alpha_t$ in the final linear linear combination $\sum_t\alpha_t h_t$ are all equal to~1 in SmoothBoost. In contrast, AdaBoost uses $\alpha_t=\frac{1}{2}\ln((1-\eps_t)/\eps_t)$, hence the quantum algorithm's approximation errors in $\eps_t$ lead to approximation errors in $\alpha_t$ that need to be kept under control.
\end{itemize}
In addition to exploiting these ``classical'' advantages in order to obtain a simpler and faster quantum booster, we also give an improved procedure to generate quantum examples over the $m$ examples from~$S$. This procedure assumes access to a non-normalized version of $D^t$ and doesn't have to worry about the normalizing factor. As explained in Section~\ref{ssec:prepqexample}, in a way quantum mechanics will take care of the proper normalization for us.\footnote{We also generate the quantum examples exactly, while \cite{arunachalam&maity:qboosting} only generate them approximately and hence has to deal with the way the errors in this process affect the other parts of their boosting procedure. Our example-generating procedure could also be used to improve the bounds of quantum AdaBoost~\cite{arunachalam&maity:qboosting}, though the result won't be as efficient as our quantum SmoothBoost. Note that if we want to use Quantum SmoothBoost with a classical weak learner $\mathcal{W}$, we can just measure the $W$ quantum examples in the computational basis to obtain the $W$ classical examples distributed according to $D^t$ that such a $\mathcal{W}$ needs as input. This still gives a speed-up in terms of the desired generalization error~$\eps$ compared to  classically generating those $W$ examples.} 

We obtain the following upper bound on the complexity of our Quantum SmoothBoost:\footnote{This bound is when we aim at constant generalization error $\eps=1/3$. We also make explicit the complexity for much smaller $\eps$ (see Theorem~\ref{th:qboostmain}).}
\begin{equation}\label{eq:qsmoothboostcomplexity}
\tilde{O}\left(\frac{W}{\gamma^4} + \frac{\sqrt{d}}{\gamma^5}\right).
\end{equation}
This improves over the complexity of the booster of \cite{arunachalam&maity:qboosting} (as given in Eq.~\eqref{eq:qadaboost})  in terms of the parameters $W$ and (especially)~$\gamma$. The $\gamma$-dependence is still worse than classical AdaBoost (as given in Eq.~\eqref{eq:adaboostcomplexity}), but not by large powers anymore. It is an interesting open question whether this $\gamma$-dependence can be improved further.

\subsection{Related work}

Our main sources of inspiration for this paper were the quantum AdaBoost of Arunachalam and Maity~\cite{arunachalam&maity:qboosting} and classical SmoothBoost of Servedio~\cite{servedio2003smooth}, and we have tried in this paper to combine the best elements of both.

Here we mention a number of related quantum papers.
Wang et al.~\cite{wangea:qboosting} (which preceded \cite{arunachalam&maity:qboosting}) give a quantum speed-up for a specific subtask of AdaBoost, namely to compute the coefficients $\alpha_t=\frac{1}{2}\ln((1-\eps_t)/\eps_t)$ that combine given base classifiers $h_1,\ldots,h_T$ into a good hypothesis  $h=sign(\sum_t \alpha_t h_t)$.
These weights $\alpha_t$ are approximated more efficiently than is possible classically using a version of approximate quantum counting. 
This, however, assumes the base classifiers have already been generated and sidesteps the most important aspect of AdaBoost, which is to generate the $h_t$'s adaptively by running the weak learner on a distribution $D^t$ that depends on $h_1,\ldots,h_{t-1}$.
The even earlier paper by Schuld and Petruccione~\cite{schuld&petruccione:qensembles} considers quantum ensembles of classifiers (rather than the linear combinations used in boosting) and runs AdaBoost as a subroutine, but does not give a quantum boosting algorithm.

AdaBoost may be viewed as an instance of the multiplicative weights update method, see for instance the presentation in~\cite[Section~3.6]{arora2012MultiplicativeWeightsAlg}.
There have been several quantum speed-ups for multiplicative weights methods in other contexts, particularly the quantum SDP-solvers of Brand\~{a}o et al.~\cite{brandao2016QSDPSpeedup,apeldoorn2017QSDPSolvers,brandao2017QSDPSpeedupsLearning,apeldoorn2018ImprovedQSDPSolving}, and the very recent quantum version of the hedge algorithm of Hamoudi et al.~\cite{hamoudiea:qhedge}.
However, none of those speed-ups for versions of multiplicative weights seems directly applicable to our boosting setting.

\section{Preliminaries}

\subsection{PAC learning}\label{sec:paclearning}

In this section we give a brief introduction to the PAC learning framework, which provides theoretical guarantees on learnability. The textbook by Shalev-Shwartz and Ben-David~\cite{ShalevShwartz2014UnderstandingML} provides an excellent and detailed introduction to the topic of classical PAC learning. 

To formally introduce the PAC learning framework, let $\mathcal{D}$ denote a probability distribution over the set of \emph{points} $\mathcal{X}$. We want to learn an unknown \emph{target function} $f:\mathcal{X}\to\mathcal{Y}$.
We will assume here that the set of labels $\mathcal{Y}$ is just $\{-1,1\}$, so we are dealing with binary classification. 
A typical situation to keep in mind is the important special case of learning Boolean functions, where $\mathcal{X}=\01^n$, or $\mathcal{X}=\cup_{n\geq 1}\01^n$.

Learning begins by choosing a learning algorithm (a ``learner'') with an associated \emph{hypothesis class} $\mathcal{H}$ of functions $h:\mathcal{X}\to\{-1,1\}$. 
This hypothesis class could be any set of functions, but good examples to keep in mind are cases where $\mathcal{X}=\01^n$ and $\mathcal{H}$ consists of objects with bounded computational power, for instance all Boolean circuits of at most a certain size, all neural networks with a specific depth and number of nodes, or all decision trees of at most a certain depth.
We will assume that each $h\in\mathcal{H}$ has a succinct description and that we can efficiently evaluate a given $h$ on a given $x\in\mathcal{X}$. For simplicity we assume such an evaluation has unit cost.

The learner is given access to a \emph{sample} $S = ((x_1, y_1),\ldots,(x_m, y_m))$, which is the training data. 
The points $x_i$ are i.i.d.\ generated according to an unknown distribution $\mathcal{D}$ on $\mathcal{X}$, and the labels $y_i=f(x_i)$ are determined by the target function $f$ that we are trying to learn. 
The learner's goal is to find an $h\in\mathcal{H}$ that fits well with the given training data, in the hope that this $h$ will generalize well to points that were not part of the data, in the sense of mostly giving the same labels as the target function. The PAC learning framework is a distribution-free setting, so we would like to design a learner that works well for every $\mathcal{D}$, in the sense of outputting a hypothesis with low generalization error.

\begin{definition}
The generalization error of $h:\mathcal{X}\to\mathcal{Y}$  w.r.t.\ target function $f:\mathcal{X}\to\mathcal{Y}$ under distribution $\mathcal{D}$ is 
\[
err(h,f,{\mathcal{D}}) = \Pr_{x\sim \mathcal{D}}[h(x)\neq f(x)].
\]
\end{definition}

Generalization error is often referred to as the \emph{true} error; it is the quantity the learner is really trying to minimize over the class $\mathcal{H}$ of available hypotheses. 

%

As the distribution $\mathcal{D}$ is anyway unknown, the generalization error of a hypothesis $h$ cannot be calculated and the learner uses the \emph{empirical} error of a hypothesis~$h$ (the fraction of the sample that $h$ mislabels) to measure its performance, as a proxy for the generalization error.

\begin{definition}
The empirical error of $h:\mathcal{X}\to\mathcal{Y}$ w.r.t.\ sample $S = ((x_1, y_1),\ldots,(x_m, y_m))$ is
\[
\hat{err}(h,S) = \Pr_{i\in_R [m]}[h(x_i) \neq y_i],
\]
where $i\in_R [m]$ means that $i$ is taken uniformly at random from $[m]=\{1,\ldots,m\}$.
\end{definition}


\begin{definition}
An $(\eps,\delta)$-PAC learner for a concept class $\mathcal{C}$ with hypothesis class $\mathcal{H}$ and sample complexity $m$, is an algorithm $\mathcal{A}$ such that the following holds for all target functions $f\in\mathcal{C}$ and all distributions $\mathcal{D}$ on $\mathcal{X}$:
\begin{itemize}
\item $\mathcal{A}$ takes as input $m$ examples $(x_1,f(x_1)),\ldots,(x_m,f(x_m))$ where the $x_i$ are i.i.d.\ according to~$\mathcal{D}$.
\item $\mathcal{A}$ outputs an $h\in\mathcal{H}$ which is ``Probably Approximately Correct'' in the sense that
$$
\Pr[err(h,f,{\mathcal{D}})\leq\eps]\geq 1-\delta,
$$
where the probability is taken over the sample and over the learner's internal randomness.
\end{itemize}
\end{definition}

The end goal is to find a learner with small sample complexity~$m$, small error probability~$\delta$, and (most important of all) small generalization error~$\eps$. Often we will start, however, with a ``weak'' learner, one whose generalization error is only slightly better than random.
Since we restricted to binary labels ($\mathcal{Y}=\{-1,1\}$), generalization error $\eps=1/2$ is no better than random guessing. A weak learner is a learner that does slightly better than that:

\begin{definition}[Weak learning]
A $\gamma$-weak learner $\mathcal{W}$ for concept class $\mathcal{C}$ with hypothesis class $\mathcal{H}_{weak}$ is a $(1/2-\gamma,0)$-PAC learner.
Hypotheses returned by a weak learner are called \emph{base classifiers}. 
\end{definition}

Following Servedio~\cite{servedio2003smooth}, we assume the weak learner~$\mathcal{W}$ has error probability $\delta=0$, so it always outputs a hypothesis with generalization error $\leq 1/2-\gamma$.
If instead we start with a $\mathcal{W}$ that has non-zero error probability, say $1/3$, then we can reduce this error probability to small $\delta>0$ as follows. Run $\mathcal{W}$ a total number of  $r=\lceil\log_3(1/\delta)\rceil$ times, each time with fresh independent examples. Then, except with probability $\leq (1/3)^r\leq\delta$, at least one of the returned hypotheses $h_1,\ldots,h_r\in\mathcal{H}_{weak}$ will have generalization error $\leq 1/2-\gamma$. However, finding (with success probability $\geq 1-\delta$) among these $r$ hypotheses one with such low error has a cost. Deciding (with success probability $\geq 2/3$) for a given hypothesis~$h$ whether it has error $\leq 1/2-\gamma$ under a given distribution can be done by sampling  $O(1/\gamma^2)$ examples according to that distribution and estimating the fraction of examples where $h$ predicts the label correctly. Searching over the $r$ hypotheses to find a good one adds a factor of $O(r)$ to the classical cost, and reducing the overall error probability from $1/3$ to $\delta$ adds another factor of $O(\log(1/\delta))$.%
\footnote{In the quantum case the $O(1/\gamma^2)$ can be replaced by $O(1/\gamma)$ using quantum approximate counting (Theorem~\ref{th:approxcount} below), and the $O(r)$ can be replaced by $O(\sqrt{r})$ using Grover's algorithm~\cite{grover:search}.}

Suppose we ideally want to run an errorless weak learner $T$ times, namely once in each of $T$ iterations. But instead we start with a weak learner with error probability~$1/3$. Reducing $1/3$ to $\delta\ll 1/T$ allows us to take a union bound over all $T$ iterations, and conclude that with high probability each of the $T$ iterations produces a base classifier with generalization error $\leq 1/2-\gamma$ (w.r.t.\ the distribution $D^t$ of that iteration).
Because here we assumed our weak learner has no error probability from the start, we do not have to factor in the additional cost for this error reduction, but it anyway doesn't significantly affect the complexities of classical or quantum boosting (Eqs.~\eqref{eq:adaboostcomplexity} and \eqref{eq:qsmoothboostcomplexity} respectively).

\bigskip


\subsection{How many examples suffice to ensure small generalization error?} 
The number of examples that are necessary and sufficient for learning is governed by the VC-dimension of the relevant hypothesis class and by the desired generalization error, as follows.
A set $S\subseteq\mathcal{X}$ of $d$ points is said to be \emph{shattered} by $\mathcal{H}$ if for each of the $2^d$ labelings $\ell:S\to\01$, there exists an $h \in\mathcal{H}$ that agrees with $\ell$ on the points in $S$. The \emph{VC-dimension} of a hypothesis class $\mathcal{H}$ is the size of a largest $S$ that is shattered by $\mathcal{H}$. Intuitively, if the VC-dimension of $\mathcal{H}$ is small, then it should be relatively simple to find a good hypothesis in it, i.e., one that minimizes empirical error. 

The following theorem implies that for sufficiently large~$m$,
\emph{every $h\in\mathcal{H}$} has a generalization error that is only slightly worse than its empirical error. Such a result means it suffices to look for a hypothesis with small empirical error. 

\begin{theorem}[Theorem~2.5 in \cite{schapire2013boosting}]\label{thm:generalizationbound}
Let $\mathcal{H}$ be a hypothesis class of finite VC-dimension $d$. Assume that a sample $S$ of size $m$ is chosen for some target function $f$, i.i.d.\ according to some distribution~$\mathcal{D}$. 
Then for every $\eta > 0$ it holds that 
\[
\Pr[\exists h \in \mathcal{H}: err(h,f,\mathcal{D}) > \hat{err}(h,S) + \eta] \leq 8\left(\frac{em}{d}\right)^d\exp{\frac{-m\eta^2}{32}}.
\]
\end{theorem}

\noindent
If we set $\eta=\eps/2$ and 
\[
m=O\left(\frac{d\log(d/(\delta\eps))+\log(1/\delta)}{\eps^2}\right),
\] 
with a sufficiently large constant in the $O(\cdot)$,
then (except with probability $\delta$), each $h\in\mathcal{H}$ has a generalization error that is at most $\eps/2$ bigger than its empirical error.
Accordingly, if a learner now outputs any hypothesis $h\in\mathcal{H}$ whose empirical error is $\leq \eps/2$, then its generalization error will be $\leq\eps$, as desired.

\subsection{Quantum PAC learning and helpful quantum  subroutines}\label{ssec:qexample}

In order to introduce the quantum boosting algorithm in Section~\ref{sec:qsmoothBoost}, we explain the query model that the quantum algorithm works with. We say that an algorithm has \emph{query access} to a string $z \in Z_a^N$ over alphabet $Z_a=\{0,\ldots,a-1\}$ if it can apply a unitary $O_z$ such that 
\[
O_z: \ket{i,b} \mapsto \ket{i, b\oplus x_i},
\]
where $i \in \01^{\lceil\log N\rceil}$, $b\in Z_a$, and $\oplus$ denotes addition modulo~$a$. Naturally, a quantum algorithm can apply $O_z$ on a superposition of distinct inputs $i$.

In the classical setting we assumed a learner is given a sample $S=((x_1,y_1), \ldots, (x_m, y_m))$ of $m$ labeled examples. Here points $x_i \in \mathcal{X}$ are independently drawn from an unknown distribution $\mathcal{D}$, and labeled $y_i=f(x_i)$ according to an unknown target function~$f$. Such a classical sample will still be the  starting point of our quantum boosting algorithm; we assume the learner has query access to the sample (viewed as a string $z\in(\mathcal{X}\times\{-1,1\})^m$).
One may think of the sample as being stored in a quantum-accessible classical memory, sometimes called QRAM. However, our setting also encompassed the case of synthetic data, where we would have an efficient procedure which, on input~$i$, computes the example $(x_i,y_i)$.

Even though the initially given sample is classical, like Arunachalam and Maity~\cite{arunachalam&maity:qboosting} we will set up our quantum booster so that it can work to improve a classical weak learner but also to improve a \emph{quantum} weak learner. The latter is given \emph{quantum examples} w.r.t.\ distribution~$D$: 
\[
\sum_{x\in\mathcal{X}} \sqrt{D(x)}\ket{x,f(x)}.
\]
One can think of a quantum example as the coherent version of a random example $(x,f(x))$ where $x\sim D$.
A quantum learner is given access to several copies of the quantum example and performs a POVM measurement, where each outcome is associated with a hypothesis~$h$ in its hypothesis class. It won't matter for the purposes of this paper, but \cite{arunachalam:optimalpaclearning} proved that in the general PAC and agnostic learning settings, the required number of classical and quantum examples are the same up to contant factor.

In the case of boosting, the weak learner will be fed quantum examples w.r.t.\ a distribution~$D$ that only has support on the $m$ given examples.
Since our initially given sample is classical, our boosting algorithm will itself have to prepare the quantum examples that it wants to feed into the weak learner in each iteration, and we have to (and will) account for the cost of this.

We also assume that we can evaluate a given $h$ (in the weak learner's hypothesis class $\mathcal{H}_{weak}$) in superposition, meaning we can apply a unitary that maps $\ket{h}\ket{x}\ket{b}\mapsto\ket{h}\ket{x}\ket{h(x)\cdot b}$; here the basis states of the first space are the names of the $h\in{\cal H}_{weak}$, the basis states of the second space are the elements of $\mathcal{X}$, and the basis states of the third space are the labels in $\mathcal{Y}=\{-1,1\}$.

The definitions of PAC learning and weak learning straightforwardly generalize to the quantum setting. We refer to the survey~\cite{arunachalam:quantumlearningsurvey} for more on this model. 
%
%
%
The following basic quantum subroutines can be derived from Brassard et al.~\cite{bhmt:countingj} (or from~\cite{aaronson&rall:qcounting} if one wants to avoid use of the quantum Fourier transform):

\begin{theorem}[Amplitude amplification]\label{th:amplampl}
Suppose we have an $m$-qubit unitary $U$ such that
$$
U\ket{0^m}=\sqrt{a}\ket{\phi_0}\ket{0}+\sqrt{1-a}\ket{\phi_1}\ket{1},
$$
and we know a lower bound $a'$ on~$a$.
Then there exists a quantum algorithm $V$ using $O(1/\sqrt{a'})$ applications of $U$ and $U^\dagger$, and $\tilde{O}(1/\sqrt{a'})$ other gates, such that
$$
V\ket{0^m}=\sqrt{b}\ket{\phi_0}\ket{0}+\sqrt{1-b}\ket{\phi_1}\ket{1},
$$
where $b\in [1/2,1]$.
\end{theorem}

\begin{theorem}[Approximate counting]\label{th:approxcount}
Suppose we have query access to a string $z\in[0,1]^N$, with sum $s=\sum_{i=1}^N z_i\geq 1$. There exists a quantum algorithm that uses $O(\frac{1}{\eps}\sqrt{N}\log(1/\delta))$ queries and $\tilde{O}(\frac{1}{\eps}\sqrt{N}\log(1/\delta))$ other operations, 
and that outputs (except with probability $\leq\delta$) an  $\tilde{s}$ such that $(1-\eps)s\leq \tilde{s} \leq(1+\eps)s$.
\end{theorem}

\section{SmoothBoost}\label{sec:smoothboost}

We first consider the classical SmoothBoost algorithm of Servedio~\cite{servedio2003smooth}. It generates only smooth distributions, in the sense that none of the examples get too much weight. In the next section we will introduce a quantum version of SmoothBoost.

We give the pseudocode of SmoothBoost in Algorithm \ref{alg:smoothboost}. There are a few cosmetic changes compared to the pseudocode of~\cite{servedio2003smooth} that will make it easier for us to quantize it later.
The algorithm takes four inputs. 
First, a $\gamma$-weak learner $\mathcal{W}$ with associated hypothesis class $\mathcal{H}_{weak}$ and cost and sample complexity~$W$. Second, a sample $S = ((x_1, y_1), (x_2, y_2), \ldots, (x_m, y_m)) \in (\mathcal{X} \times \{-1,1\})^m$, for some sample size~$m$ that we will choose later. Lastly, a parameter $\kappa \in (0,1)$ which controls the empirical error of SmoothBoost and a parameter $\theta \in [0, \frac{1}{2})$ which controls the desired margin of the output hypothesis $h$. The goal of SmoothBoost is to output a hypothesis $h: \mathcal{X} \to \{-1, 1\}$ with small empirical error (and as we shall see later, for sufficiently large $m$, this $h$ will also have large generalization error). The final $h$ is going to be the sign of a sum of elements of $\mathcal{H}_{weak}$, so the strong learner's hypothesis class is larger than that of the weak learner.

\begin{algorithm}[!h]
\caption{SmoothBoost}\label{alg:smoothboost}
\begin{algorithmic}[1]

\Require A $\gamma$-weak learner $\mathcal{W}$ with complexity $W$.

A sample $S = ((x_1, y_1), \ldots, (x_m, y_m)) \in (\mathcal{X} \times \{-1, 1\})^m$. 

Parameters $\kappa \in (0,1), \theta \in [0, \frac{1}{2})$.
\vspace{0.15cm}
\Ensure Hypothesis $h: \mathcal{X} \to \{-1,1\}$.

\vspace{0.15cm}

\Function{SmoothBoost}{$\mathcal{W}$, $S$, $\kappa$, $\theta$}

\State For all $i \in [m]$ initialize: $N_i^0 \gets 0, \; M_i^1 \gets 1$.

\State $t \gets 1$ 

\While{true}
 \State\label{smoothboost:approxsum} Compute $s=\sum_{i=1}^m M^t_i$.

\State\label{smoothboost:condition} If $s<\kappa m$ then $T \gets t-1$, \textbf{return} $h\gets sign(\sum_{t=1}^T h_t)$, and \textbf{terminate}.

\State\label{smoothboost:prepareexamples} Prepare $W$ i.i.d.\ examples w.r.t.\ distribution $D^t=M^t/\sum_i M^t_i$ (see Footnote~\ref{footnote:rejectionsampling}).

\State\label{smoothboost:weaklearner} 
Feed those $W$ examples into the weak learner $\mathcal{W}$ to obtain base classifier $h_t$. 
\State\label{smoothboost:weightupdate} For all $i \in [m]$ set
$N^t_i \gets N^{t-1}_i + h_t(x_i)y_i - \theta$
and
$
M^{t+1}_i \gets
\begin{cases}
1 & \text{for } N^t_i < 0 \\ 
(1 - \gamma)^{\frac{N^t_i}{2}} & \text{for } N^t_i \geq 0 
\end{cases}$

\State $t\gets t+1$.
\EndWhile
\EndFunction
\end{algorithmic}
\end{algorithm}

The central objects in this algorithm are the vectors $M^1,M^2,\ldots,M^T\in[0,1]^m$, which are unnormalized distributions over the $m$ examples. 
The distribution~$D^t$ is the normalized version of $M^t$.
SmoothBoost starts by initializing weights to $N^0_i = 0$ and $M^1_i = 1$, for all $i \in [m]$, so $D^1$ is uniform.
In each iteration, Step~\ref{smoothboost:condition} checks whether the sum of $M^t_i$ is below $\kappa m$, and if so it terminates.
Otherwise it runs the weak learner on $W$ i.i.d.\ examples sampled (from $S$) according to $D^t$,
producing a base classifier $h_t:\mathcal{X} \to \{-1,1\}$. Step~\ref{smoothboost:weightupdate} updates $N^{t-1}$ to $N^t$ and $M^t$ to $M^{t+1}$. For each $i \in [m]$, $N^t_i$ is the cumulative amount by which hypotheses $h_1, \ldots, h_t$ beat the desired margin $\theta$. If $x_i$ got correctly classified by $h_t$, then it will get higher weight $N^t_i$, which results in smaller weight $M^{t+1}_i$ and smaller probability $D^{t+1}_i$ in the next round of boosting. This mechanism forces the next run of the weak learner $\mathcal{W}$ to ``zoom in'' (i.e., assign higher probabilities) to systematically misclassified instances. The procedure terminates if the sum of all weights $\sum_{i \in [m]} M^t_i$ gets sufficiently small, as controlled by the parameter $\kappa$.

We now state three results from Servedio~\cite{servedio2003smooth} that show, respectively, that the intermediate distributions are smooth, that SmoothBoost terminates after a small number of iterations, and that it returns a hypothesis with low empirical error.

\begin{claim}[Lemma~1 of \cite{servedio2003smooth}]\label{smoothboost:smooth}
For each $1 \leq t \leq T$, it holds that $\max_{i \in m} | D^t_i | \leq \frac{1}{\kappa m}$.
\end{claim}

\begin{proof}
This follows from the condition of Step~\ref{smoothboost:condition}: before termination we have $\sum_{i=1}^m M_i^t\geq\kappa m$ and $M^t_i\in[0,1]$, hence $D^t_i=M_i^t/\sum_{i\in[m]} M_i^t\leq 1/\kappa m$ for all $i\in[m]$.
\end{proof}

\begin{claim}[Theorem~3 of \cite{servedio2003smooth}]\label{smoothboost:runningtime}
If $\theta = \frac{\gamma}{2+\gamma}$ and for all $t$ it holds that $\Pr_{i\sim D^t}[h_t(x_i)\neq y_i] \leq \frac{1}{2} - \gamma$, then SmoothBoost terminates with $T < \frac{2}{\kappa \gamma^2 \sqrt{1 - \gamma}}$ iterations.  
\end{claim}

\begin{proof}
\cite[Lemmas~4 and~5]{servedio2003smooth}
imply $\frac{2m}{\gamma\sqrt{1-\gamma}}>\gamma\sum_{t=1}^T\sum_{i=1}^m M_i^t$ (this is the hard part of Servedio's correctness proof of SmoothBoost). We have $\sum_{i=1}^m M_i^t\geq \kappa m$ for all $t$ until termination. Hence $\frac{2m}{\gamma\sqrt{1-\gamma}}>\gamma T \kappa m$, which implies the claim.
\end{proof}

\begin{claim}[Theorem~2 of \cite{servedio2003smooth}]\label{smoothboost:empiricalerror}
After $t$ iterations of SmoothBoost, the hypothesis $h=sign(\sum_{t=1}^t h_t)$ has empirical error $\hat{err}(h) \leq \sum_{i=1}^m M^{t+1}_i/m$.
\end{claim}

\begin{proof}
Note that $N_i^T=\sum_{t=1}^T (h_t(x_i)y_i-\theta)$. Hence if $i$ is such that $\sum_{t=1}^T h_t(x_i)y_i<\theta T$, then $N_i^T<0$ and $M^{T+1}_i=1$.
The final hypothesis $h$ errs on the $i$th example iff $\sum_{t=1}^T h_t(x_i)y_i<0$.
We upper bound the number of $i\in[m]$ for which this happens:
\[
\left|\left\{i\mid \sum_{t=1}^T h_t(x_i)y_i<0\right\}\right|
\leq \left|\left\{i\mid \sum_{t=1}^T h_t(x_i)y_i<\theta T\right\}\right|
=\sum_{i:\sum_{t=1}^T h_t(x_i)y_i<\theta T}M_i^{T+1}
\leq\sum_{i=1}^m M_i^{T+1}.
\]
\end{proof}

\noindent
Since SmoothBoost terminates if $\sum_{i=1}^m M_i^{T+1}<\kappa m$, Claim~\ref{smoothboost:empiricalerror} implies that the empirical error of the final hypothesis is $<\kappa$.

Combining the previous two claims, we see that the empirical error decreases like $O(1/(T\gamma^2))$. This contrasts with AdaBoost, where the empirical error goes down exponentially fast in $T$ and hence can be driven down to $<1/m$ (and hence to~0) quite cheaply.
SmoothBoost does not drive the empirical error down all the way to~0, because that would require setting $\kappa<1/m$ which implies a very large number of iterations, $T=O(m/\gamma^2)$. However, small but non-zero empirical error is good enough for our purposes, because (with sufficiently large sample size~$m$) that already implies small generalization error.

We will choose $\kappa$ to be $\eps/2$, which sets the above upper bound on the empirical error of the final hypothesis~$h$ to half of the allowed generalization error. By the discussion following Theorem~\ref{thm:generalizationbound}, if the sample size~$m$ is large enough, the final hypothesis will have generalization error $\leq\eps$, as desired. The required~$m$ depends on the VC-dimension of the hypothesis class $\mathcal{H}_{strong}$ of SmoothBoost, which consists of signs of sums of $T$ elements of the hypothesis class $\mathcal{H}_{weak}$ of the weak learner. The VC-dimensions of these two classes are related as follows:

\begin{claim}[Shalev-Shwartz \& Ben-David, p.~109 \cite{ShalevShwartz2014UnderstandingML}]\label{claim:adaboostvc}
Let $\mathcal{H}_{weak}$ be a hypothesis class of VC-dimension $d$ and $\mathcal{H}_{strong}=\{sign(\sum_{i=1}^T h_i)\mid h_1,\ldots,h_T\in\mathcal{H}_{weak}\}$.
Then the VC-dimension of $\mathcal{H}_{strong}$ is $D=O(Td\log(Td))$.
\end{claim}

By Theorem~\ref{thm:generalizationbound} and the fact that $T=O(\frac{1}{\eps\gamma^2})$ it thus suffices to take 
\begin{equation}\label{eq:mbound}
m=O\left(\frac{D\log(D/(\delta\eps))+\log(1/\delta)}{\eps^2}\right)
=O\left(\frac{d\log(d/(\delta\eps\gamma))^2}{\eps^3\gamma^2}+\frac{\log(1/\delta)}{\eps^2}\right)
\end{equation}
examples in order to be able to infer (with success probability $\geq 1-\delta$) generalization error $\leq\eps$ from empirical  error $\leq\eps/2$.

Finally, let us determine the complexity of SmoothBoost, in terms of the overall number of elementary operations and queries to the sample and to the $h_t$. There are $T=O(\frac{1}{\eps\gamma^2})$ iterations. Each iteration involves one application of the weak learner $\mathcal{W}$, and $\tilde{O}(m)$ other operations. The weak learner needs to be fed $W$ examples sampled according to distribution $D^t$. Using rejection sampling, we can generate $W$ such examples at cost $O(W/\kappa)=O(W/\eps)$.\footnote{Specifically, Step~\ref{smoothboost:prepareexamples} of SmoothBoost can be implemented as follows.
Sample $i\in[m]$ uniformly. With probability $M_i^t$ output $(x_i,y_i)$, and otherwise repeat. Since the probability to output $(x_i,y_i)$ is proportional to $M_i^t$, the example (if we indeed output an example) is sampled according to the desired probability distribution~$D^t$. Note that the probability that we output an example in one try is $\frac{1}{m}\sum_i M_i^t \geq \kappa$, because of the condition of Step~\ref{smoothboost:condition}. Hence the expected number of repetitions before we output an example is $\leq 1/\kappa$.\label{footnote:rejectionsampling}}

\begin{theorem}
Let $\mathcal{W}$ be a $\gamma$-weak learner of complexity $W$ for concept class $\mathcal{C}$, with hypothesis class $\mathcal{H}_{weak}$ of VC-dimension~$d$. Then given $m$ examples according to Eq.~\eqref{eq:mbound}, SmoothBoost is an $(\eps,\delta)$-PAC learner for $\mathcal{C}$, with hypothesis class $\mathcal{H}_{strong}$. It runs the weak learner $O(\frac{1}{\eps\gamma^2})$ times and uses 
\[
\tilde{O}(T(W/\eps+m))
=\tilde{O}\left(\frac{W}{\eps^2\gamma^2} + \frac{m}{\eps\gamma^2}\right)
=\tilde{O}\left(\frac{W}{\eps^2\gamma^2} + \frac{d}{\eps^4\gamma^4}\right)
\]
other operations (elementary computational steps, queries to the sample, and evaluations of base classifiers).
\end{theorem}

\section{Quantum Smooth Boosting}\label{sec:qsmoothBoost}

In this section we introduce our quantum version of SmoothBoost. The algorithm is given query access to a quantum (or classical) weak learner $\mathcal{W}$ with sample complexity $W$, and to a sample $S$ of size~$m$. 
The quantum weak learner needs to be fed \emph{quantum} examples according to the distribution $D^t$ obtained by normalizing the weight-vector~$M^t$. We will start with that.

\subsection{Preparing quantum examples}\label{ssec:prepqexample}

Here we show how we can efficiently prepare quantum examples w.r.t.\ the distribution $D^t$ induced by the non-normalized $M^t$, thanks to its smoothness.
This may be viewed as a quantum analogue of the classical rejection sampling sketched in Footnote~\ref{footnote:rejectionsampling}.

\begin{theorem}\label{th:qexample}
Suppose we have query access to the $m$ numbers $M_1,\ldots,M_m\in[0,1]$.
Let $s=\sum_{i=1}^m M_i$ be their (unknown) sum, which has a known lower bound of $\kappa m$. 
Define a probability distribution $D$ on $[m]$ by $D_i=M_i/s$.
Then using an expected number of $O(1/\sqrt{\kappa})$ queries and $\tilde{O}(1/\sqrt{\kappa})$ other gates, we can  prepare the state
\[
\sum_{i=1}^m  \sqrt{D_i} \ket{i}.
\]
\end{theorem}

\begin{proof}
Start by preparing the uniform state
\[
\frac{1}{\sqrt{m}}\sum_{i=1}^m\ket{i}\ket{0}.
\]
Using two queries (the second to uncompute the value $M_i$), and a few other gates to implement a conditional rotation by angle $\arcsin(\sqrt{M_i})$, prepare
\[
\frac{1}{\sqrt{m}}\sum_{i=1}^m\ket{i} (\sqrt{M_i}\ket{0} + \sqrt{1 - M_i}\ket{1}).
\]
The squared norm of the part of the state ending in $\ket{0}$ is $s/m\geq \kappa$. Now use $O(1/\sqrt{\kappa})$ rounds of amplitude amplification (Theorem~\ref{th:amplampl}) to increase that squared norm to $\geq 1/2$. 
This costs $O(1/\sqrt{\kappa})$ queries and $\tilde{O}(1/\sqrt{\kappa})$ other gates.

If we measure the last qubit of the resulting state, then we obtain outcome~0 with probability $\geq 1/2$ and the state collapses to the state that we want to prepare (with an extra $\ket{0}$-qubit that we can remove).
Note that we know when we succeed to produce the desired state. Since the probability of success is $\geq 1/2$, the expected number of repetitions before success is $\leq 2$.
\end{proof}

\noindent
Once we have produced a copy of the state
\[
\sum_{i=1}^m \sqrt{D_i} \ket{i}.
\]
we can easily convert this into a quantum example
\[
\sum_{i=1}^m \sqrt{D_i} \ket{x_i,y_i}.
\]
by querying the sample~$S$.

\subsection{Quantizing SmoothBoost}

The pseudocode of Quantum Smooth Boosting is given in Algorithm~\ref{alg:quantumsmoothboost}. The algorithm receives as input a weak quantum learner $\mathcal{W}$ with sample complexity $W$, and query access to a sample $S$ of $m$ examples. Additionally, it receives two parameters $\kappa,\theta$. 

The algorithm looks a bit different from classical SmoothBoost because it doesn't update the $m$-dimensional vectors $N^t_i$ and $M^t_i$ explicitly anymore; the $O(m)$ that this costs is more than we are willing to spend in the quantum case.
Instead, we will store the earlier base classifiers $h_1,\ldots,h_t$. Queries to these classifiers together with queries to the sample~$S$ allow us to calculate each entry $N^t_i$ and $M^t_i$ on demand in time $\tilde{O}(t)$, via the formulas of Step~\ref{smoothboost:weightupdate} of SmoothBoost. 

The algorithm begins by initializing $N^0$, $M^1$ and by setting $t = 1$. Like in the classical case, we iterate until the sum of the weights $\sum_{i \in[m]} M^t_i$ becomes small enough. In contrast to the classical case, we do not have the time to sum these $m$ numbers exactly, so we will instead estimate the sum with small approximation error using quantum counting (Theorem~\ref{th:approxcount}).

\begin{algorithm}[hbt]
\caption{Quantum SmoothBoost}\label{alg:quantumsmoothboost}
\begin{algorithmic}[1]

\Require A $\gamma$-weak quantum learner $\mathcal{W}$ with complexity $W$.

A sample $S = ((x_1, y_1), \ldots, (x_m, y_m)) \in (\mathcal{X}\times \{-1, 1\})^m$. 

Parameters $\kappa \in (0,1), \theta \in [0, \frac{1}{2})$.
\vspace{0.15cm}
\Ensure Hypothesis $h: \mathcal{X} \to \{-1,1\}$.

\vspace{0.15cm}
\Function{QuantumSmoothBoost}{$\mathcal{W}$, $S$, $\kappa$, $\theta$}
\State $t \gets 1$

\While{true}
 \State\label{quantumsmoothboost:approxsum} Compute an estimate $\tilde{s}$ of $s=\sum_{i=1}^m M^t_i$ with multiplicative error 1.1 (using Theorem~\ref{th:approxcount}),
 \hspace*{3.5em}where the $M^t_i$ are as defined in Step~\ref{smoothboost:weightupdate} of SmoothBoost (and only computed on demand).

\State\label{quantumsmoothboost:terminate} If $\tilde{s}<\kappa m$ then $T \gets t-1$, \textbf{return} $h\gets sign(\sum_{t=1}^T h_t)$, and \textbf{terminate}.


\State\label{quantumsmoothboost:preparestate} Prepare $W$ copies of example $\ket{D^t}$ w.r.t.\ distribution 
$D^t_i = \frac{M^t_i}{\sum_{i \in[m]} M^t_i}$ (using Theorem~\ref{th:qexample}).

\State\label{quantumsmoothboost:quantumweaklearner} Feed those $W$ examples into the weak learner $\mathcal{W}$ to obtain base classifier $h_t$. 

\State $t\gets t+1$.
\EndWhile
\EndFunction
\end{algorithmic}
\end{algorithm}

Quantum Smoothboost runs $O(TW)$ quantum subroutines that each have some error probability. By setting this error probability to be $\ll 1/TW$, the union bound implies that the probability that at least one of them will fail, is very small. The extra cost-factor $\log(TW)$ that this error-reduction incurs will be absorbed by our $\tilde{O}(\cdot)$ notation.

If we condition on the very-high-probability event that the various quantum subroutines involved all succeed, then the weights $N^t$ and $M^t$ are just equal to the weights as they would be in classical SmoothBoost with the same number of iterations.  Because our approximation $\tilde{s}$ of $s$ for the stopping criterion has small multiplicative error, the smoothness of the intermediate distributions $D^t$ before termination can be marginally worse than in classical SmoothBoost (Claim~\ref{smoothboost:smooth}):
For each $1 \leq t \leq T$, it holds that $\max_{i \in m} | D^t_i | \leq \frac{1.1}{\kappa m}$.

Quantum SmoothBoost terminates if $\tilde{s}<\kappa m$. Because $\tilde{s}$ might underestimate the true $s$ by at most a factor~1.1, upon termination we have $s<1.1\kappa m$ and hence Claim~\ref{smoothboost:empiricalerror} implies empirical error $\hat{err}(h)<1.1\kappa$. 
We set $\kappa=\eps/2.2$ in order to ensure  $\hat{err}(h)\leq \eps/2$. 
Like before, we choose the sample size~$m$ given by Eq.~\eqref{eq:mbound} to ensure generalization error $\leq\eps$.

The total number of iterations is still $O(\frac{1}{\eps\gamma^2})$.\footnote{There is one small change in the proof of  Claim~\ref{smoothboost:runningtime}: since we condition on all runs of quantum counting in Step~\ref{quantumsmoothboost:approxsum} giving an estimate of $\sum_i M_i^t$ up to multiplicative error 1.1, we now have $\sum_{i=1}^m M_i^t\geq \kappa m/1.1$ for all $t$ until termination.}
It remains to determine the complexity of one iteration. The most costly steps in one iteration are Steps~\ref{quantumsmoothboost:approxsum} and~\ref{quantumsmoothboost:preparestate}. As a subroutine these will use the fact that we can compute $M_i^t$ and $N_i^t$ using $O(t)=O(T)$ calls to the earlier base classifiers and the sample~$S$. 
\begin{enumerate}
    \item[Step~\ref{quantumsmoothboost:approxsum}.] The approximation of $s=\sum_{i=1}^m M^t_i$ up to multiplicative error 1.1  using  Theorem~\ref{th:approxcount} costs $\tilde{O}(T\sqrt{m})$ (for simplicity assume $\eps\gg 1/m$ to ensure the condition $s\geq 1$ in Theorem~\ref{th:approxcount} holds.)
    \item[Step~\ref{quantumsmoothboost:preparestate}.] Preparing one copy of $\ket{D^t}$ costs $\tilde{O}(T/\sqrt{\eps})$ by Section~\ref{ssec:prepqexample} (using our setting of $\kappa=\eps/2.2$), so overall this step costs $\tilde{O}(WT/\sqrt{\eps})$.
\end{enumerate}
Adding these costs shows that one iteration costs $\tilde{O}(T(W/\sqrt{\eps}+\sqrt{m}))$.
Plugging in $T=O(\frac{1}{\eps\gamma^2})$, and the same sample size $m=\tilde{O}(d/\eps^3\gamma^2)$ as for classical SmoothBoost (from Eq.~\eqref{eq:mbound}), gives our main result:

\begin{theorem}\label{th:qboostmain}
Let $\mathcal{W}$ be a $\gamma$-weak quantum learner of complexity $W$ for concept class $\mathcal{C}$, with hypothesis class $\mathcal{H}_{weak}$ of VC-dimension~$d$. Then given $m$ examples according to Eq.~\eqref{eq:mbound}, QuantumSmoothBoost is an $(\eps,\delta)$-PAC learner for $\mathcal{C}$, with hypothesis class $\mathcal{H}_{strong}$. It runs the weak learner $O(\frac{1}{\eps\gamma^2})$ times and uses 
\[
\tilde{O}(T^2(W/\sqrt{\eps}+\sqrt{m}))
=\tilde{O}\left(\frac{W}{\eps^{2.5}\gamma^4} + \frac{\sqrt{m}}{\eps^2\gamma^4}\right)
=\tilde{O}\left(\frac{W}{\eps^{2.5}\gamma^4} + \frac{\sqrt{d}}{\eps^{3.5}\gamma^5}\right)
\]
other operations (elementary computational steps, queries to the sample, and evaluations of base classifiers).
\end{theorem}

\noindent
For direct comparison with the quantum boosting result of Arunachalam and Maity~\cite{arunachalam&maity:qboosting}, we instantiate this by setting $\eps=\delta=1/3$,
in which case the complexity of Quantum SmoothBoost is
\[
\tilde{O}\left(\frac{W}{\gamma^4} + \frac{\sqrt{d}}{\gamma^5}\right).
\]
This polynomially improves over the time complexity $\displaystyle\tilde{O}\left(\frac{W^{1.5}\sqrt{d}}{\gamma^{11}}\right)$ of the quantum version of AdaBoost of~\cite{arunachalam&maity:qboosting}, in the $W$-dependence but especially in the $\gamma$-dependence.

\section{Future work}

This work leaves open many questions for future work:
\begin{itemize}
    \item The $\gamma$-dependence of Quantum Smoothboost is still slightly worse than in  classical boosting ($1/\gamma^5$ vs $1/\gamma^4$). Is there a way to improve this further, or can we prove a lower bound on the $\gamma$-dependence for every quantum boosting algorithm that has $\sqrt{d}$-dependence on the VC-dimension of the weak learner's hypothesis class? 
    \item Our quantum version of SmoothBoost improves the cost per iteration but not the number of iterations, which remains $T=O(1/\eps\gamma^2)$. Can we reduce the number of iterations by quantizing SmoothBoost differently, or by quantizing some other boosting approach?
    \item Boosting has many applications in theory and practice. Can we find applications where quantum Smoothboost is particularly suitable---some problem where the weak learner has relatively large advantage~$\gamma$ and large  VC-dimension~$d$, so that the square-root improvement in $d$ dominates the worse dependence on $1/\gamma$?
    \item Can we do boosting for \emph{agnostic} learning, where the label $y$ of an example ($x,y)$ is not determined by~$x$ but $(x,y)$ is jointly generated by some distribution $\mathcal{D}$ on $\mathcal{X}\times\mathcal{Y}$?
    \item What about learning with various kind of noise in the sample: random classification noise, or Massart noise, or Tsybakov noise, or malicious noise?
    Servedio~\cite{servedio2003smooth} designed SmoothBoost motivated by its ability to deal with malicious noise on the labels of the sample: 
    if $1/100$ of the $m$ given examples have their label flipped, then a distribution that puts probability $\leq c/m$ on each $i$ only puts total probability $\leq c/100$ on the erroneous examples. Servedio used this to give a learning algorithm for linear threshold functions that is robust against small, $\gamma$-dependent amounts of malicious noise (see \cite{long&servedio:moremalicious} and references therein for follow-up work). 
    \item What about learning functions that have a larger range than just $\{-1,1\}$?
\end{itemize}

\subsection*{Acknowledgements}
We thank Srinivasan Arunachalam for many helpful comments, Min-Hsiu Hsieh for sending us an updated version of~\cite{wangea:qboosting} and answering some questions about this paper, and Yassine Hamoudi for answering a question about~\cite{hamoudiea:qhedge}.

\bibliographystyle{alpha}
\bibliography{qcs}

\end{document}